\documentclass{llncs}
\usepackage{amsmath,amsfonts,amssymb,stmaryrd,natbib,cmll}
\usepackage{graphicx}
\usepackage[nohug,midshaft]{diagrams}
\usepackage[arrow, matrix, tips]{xy}
\SelectTips{cm}{10}
\DeclareMathOperator{\ob}{ob}

\newcommand{\bigcat}[1]{\mathsf{#1}}

\newcommand{\cd}[2][]{\vcenter{\hbox{\xymatrix#1{#2}}}}

\newcommand{\citegenitif}[1]{\citeauthor{#1}'s \citeyearpar{#1}}
\newcommand{\citegenitifs}[1]{\citeauthor{#1}' \citeyearpar{#1}}

\newcommand{\theory}{\mathbb{T}}
\newcommand{\C}{\mathcal{C}}
\newcommand{\D}{\mathcal{D}}

\newcommand{\SMCCat}{\bigcat{SMCCat}}

\newcommand{\ty}{\textsf{ty}}
\newcommand{\smc}{\textsc{smc}}
\newcommand{\imll}{\textsc{imll}}

\newcommand{\thg}{{\mathord{\text{--}}}}

\newcommand{\id}{\mathrm{id}}
\newcommand{\iso}{\cong}
\newcommand{\impll}{\multimap}
\newcommand{\tens}{\otimes}
\newcommand{\bigtens}{\bigotimes}
\newcommand{\Bigtens}[1]{\displaystyle{\bigtens_{#1}}}
\newcommand{\Coprod}[1]{\displaystyle{\coprod_{#1}}}

\newcommand{\nop}[1]{#1^{\bot}}
\newcommand{\name}[1]{\ulcorner #1 \urcorner}

\title{Graphical presentations of symmetric monoidal closed theories}

\author{Richard Garner\inst{1} \and Tom Hirschowitz\inst{2} \and Aur\'elien Pardon\inst{3}}
\institute{Uppsala University \and CNRS, Universit\'e de Savoie \and ENS Lyon}

\begin{document}

\maketitle
\begin{abstract}
  We define a notion of symmetric monoidal closed (\smc) theory, consisting of a \smc\
  signature augmented with equations, and describe the classifying categories of such theories in terms
  of proof nets.
\end{abstract}

\section{Introduction}

In this note, in preparation for a sequel using symmetric monoidal closed
(\smc) categories to reconstruct \citegenitif{Milner:bigraphs} bigraphs, we
define a notion of \smc\ theory, and give a graphical presentation of the free
\smc\ category generated by such a theory.

\subsection{Symmetric monoidal closed theories}

Recall that a \emph{many-sorted algebraic theory} is specified by first giving
a \emph{signature}---a set of sorts $X$ and a set $\Sigma$ of operations with
arities---together with a set of \emph{equations} over that signature. For
example, the theory for monoids is specified by taking only one sort $x$, and
operations
\begin{center}
    \hfil $m \colon x \times x \to x$ \hfil \mbox{ and } \hfil $e \colon 1 \to x$, \hfil
\end{center}
together with the usual associativity and unitality equations. We may equally
well view this signature as given by a graph
$$ \begin{diagram}
    x \times x & \rTo^{m} & x & \lTo^{e} & 1
\end{diagram} $$
whose vertices are labelled by objects of the free category with finite
products generated by $X$. In this paper, we follow the same route, but
replacing from the start finite products with \emph{symmetric monoidal closed}
structure. Thus, an \smc\ signature is given by a set of sorts $X$, together
with a graph with vertices in the free \smc\ category generated by $X$, so that
instead of cartesian product, we have available the logical connectives of
\citet{ll,Girard:survey}'s Intuitionistic Multiplicative Linear Logic
(henceforth \imll): a tensor product $\tens$, its right adjoint $\impll$, and
its unit $I$.  This permits idioms from higher-order abstract syntax
\citep{PfenningElliott:hoas}, e.g., taking the graph
\begin{equation}
    \label{eq:linsig}
    \begin{diagram}
	(x \impll x) & \rTo^{\lambda} & x & \lTo^{@} & (x \tens x)
    \end{diagram}
\end{equation}
as a signature. An \smc\ theory is now given by a \smc\ signature, together
with a set of equations over that signature. This notion of theory gives rise
to a \emph{functorial semantics} in the sense of \citet{Lawvere:thesis}, the
crux of which is the following. We may define a notion of \emph{model} for an
\smc\ theory in an arbitrary \smc\ category, and may associate to each \smc\
theory $\mathbb T$ a \emph{classifying category} $\C_{\mathbb T}$: this being a
small \smc\ category for which strict \smc\ functors $\C_{\mathbb T} \to \D$
are in bijection with models of $\mathbb T$ in $\D$. The existence of
$\C_{\mathbb T}$ follows from general considerations of categorical universal
algebra; but the description this gives of $\C_\mathbb T$ is syntactic. The
main purpose of this paper is to give a graphical presentation of $\C_\mathbb
T$. Its objects will be \imll\ formulae, while its morphisms are variants of
\citeauthor{Hughes:freestar}' \citeyearpar{Hughes:freestar} proof nets,
satisfying a correctness criterion familiar from \citegenitif{Danos:mll}.

%

\subsection{Related work}
There is an extensive literature devoted to describing free \smc\ categories of
the kind we consider here. In their seminal work on coherence for closed
categories, \citet{Kelly:coherence} introduced what are now known as
Kelly-MacLane graphs, but did not go so far as to obtain a characterisation of
free \smc\ categories. Such a construction was first carried out by
\citet{Trimble:phd}, and subsequently \citet{Cockett:weakly} and
\citet{Tan:phd}, using ideas taken from \citeauthor{ll}'s
\citeyearpar{ll,Girard:survey} proof nets (actually, \citet{Cockett:weakly}
construct the free star-autonomous category, but the free \smc\ category is
obtained as its full subcategory of \imll\ formulae). Variations on this theme
are presented by \citet{Lamarche:nets} and \citet{Hughes:freestar}. In all
cases, morphisms are roughly equivalence classes of proof nets, with variations
in the presentation. In our sequel to this paper, we wish to make use of
Hughes' presentation, mainly because:
\begin{itemize}
\item it reduces the graphical burden to the minimum: where others
    introduce nodes corresponding to linear logical connectives, Hughes
    does not;
\item its composition behaves nicely: it is defined on representatives and
    given by a straightforward gluing of graphs, where others rely on
    tricky mechanisms, e.g., \citegenitif{Trimble:phd} \emph{rewiring}.
\end{itemize}
On the other hand, Hughes' equivalence classes of proof nets have the
inconvenience of lacking normal forms, which, e.g., Trimble's enjoy.

However, Hughes only construct the free \smc\ over a \emph{set}, which
merely accounts for the sorts of a signature.  Thus we must extend his
construction to deal with an arbitrary \smc\ theory, which we do by
reducing from the general case to that of a free \smc\ on a set.
\citet{Cheng:trees} observed a relationship between trees and
Kelly-MacLane graphs, of which our result is essentially a
generalisation.

\section{Symmetric monoidal closed theories}\label{sec:theories}
Given a set $X$, we write $\overline X$ for the set of symmetric monoidal
closed (henceforth \smc) types over $X$; it is inductively generated by the
following grammar:
\begin{equation*}
    e ::= x \mathrel{|} I \mathrel{|} e \tens e \mathrel{|} e \multimap e \qquad \text{(where $x \in X$)\text.}
\end{equation*}
By a \emph{\smc\ signature}, we mean a quadruple $(X, \Sigma, s, t)$ where $X$
is a set of ground types, $\Sigma$ a set of ground terms, and $s, t \colon
\Sigma \to \overline X$ are source and target arity functions. We may also
write $\Sigma(a,b)$ for the set of $f \in \Sigma$ for which $s(f) = a$ and
$t(f) = b$. For each \smc\ signature, we inductively generate the set
$\overline \Sigma$ of \emph{derived terms}, together with source and target
functions $\overline s, \overline t \colon \overline \Sigma \to \overline X$,
as follows. We require that for each $f \in \Sigma(a,b)$, we have $f \in
\overline \Sigma(a, b)$; for each $a,b,c \in \overline X$, we have
\begin{align*}
  \alpha_{abc} & \in \overline \Sigma\big(a \otimes (b \otimes c),\, (a \otimes b) \otimes c\big)\text; &
  \alpha_{abc}^{-1} & \in \overline \Sigma\big((a \otimes b) \otimes c,\, a \otimes (b \otimes c)\big)\text; \\
  \lambda_a & \in \overline \Sigma(I \otimes a,\, a)\text; &
  \lambda_a^{-1} & \in \overline \Sigma(a,\, I \otimes a)\text; \\
  \rho_a & \in \overline \Sigma(a \otimes I,\, a)\text; &
  \rho_a^{-1} & \in \overline \Sigma(a,\, a \otimes I)\text; \\
  \sigma_{ab} & \in \overline \Sigma\big(a \otimes b,\, b \otimes a)\text; &
  \epsilon_{ab} & \in \overline \Sigma\big((a \multimap b) \otimes a,\, b\big)\\
  \text{and }\ \  \eta_{ab} & \in \overline \Sigma\big(a,\, b \multimap (a \otimes
  b)\big)\text;
\end{align*}
for each $f \in \overline \Sigma(a, b)$ and $g \in \overline \Sigma(b, c)$, we
have $g \circ f \in \overline \Sigma(a, c)$; for each $a \in \overline X$, we
have $\id_a \in \overline \Sigma(a,a)$; and for each $f \in \overline \Sigma(a,
b)$ and $g \in \overline \Sigma(c, d)$, we have $f \otimes g \in \overline
\Sigma(a \otimes b,\, c \otimes d)$ and $f \multimap g \in \overline \Sigma(c
\multimap b,\, a \multimap d)$. By an \emph{equation} over a \smc\ signature,
we mean a string of the form $u = v \colon a \to b$  for some $a, b \in
\overline X$ and $u, v \in \overline \Sigma(a, b)$; and by a \emph{syntactic
\smc\ theory} we mean an \smc\ signature $(X, \Sigma)$ together with a set $E$
of equations over it.

\begin{example}\hfill
\begin{itemize}
\item The syntactic theory of monoids has a single ground sort $x$, ground
    terms $e \in \Sigma(I, x)$ and $m \in \Sigma(x \otimes x, x)$, and
    three equations
\begin{align*}
    (m \circ (m \otimes \id_x)) \circ \alpha_{xxx} &= m \circ (\id_x \otimes m)
    \colon x \otimes (x \otimes x) \to x\\
    m \circ (e \otimes \id_x) &= \lambda_x \colon I \otimes x \to x\\
    m \circ (\id_x \otimes e) &= \rho_x \colon x \otimes I \to x\text.
\end{align*}
\item The syntactic theory of the linear lambda-calculus has a single
    ground sort $x$ and two terms, $\lambda \in \Sigma(x \multimap x, x)$
    and $@ \in \Sigma(x \otimes x, x)$. Its single equation is the
    $\beta$-rule
\begin{equation*}
    @ \circ (\lambda \otimes \id_x) = \epsilon_{xx} \colon (x \multimap x) \otimes x
    \to x\text.
\end{equation*}
\end{itemize}
\end{example}
Given a syntactic theory $\mathbb T$ and a \smc\ category $\D$, we may define a
notion of \emph{interpretation} $F \colon \mathbb T \to \D$. Such an $F$ is
given by a function $F_X \colon X \to \ob \D$ interpreting the ground types of
the theory, together with a family of functions
\begin{equation*}
    F_{a,b} \colon \Sigma(a,b) \to \D\big(\overline{F_X}(a), \overline{F_X}(b)\big)
    \qquad \text{(for $a,b \in \overline X$)}
\end{equation*}
interpreting the basic terms; here we write $\overline {F_X}$ for the unique
extension of $F_X$ to a function $\overline X \to \ob \D$ commuting with the
\smc\ type constructors. These data are required to satisfy each of the
equations of the theory, in the sense that
\begin{equation*}
    u = v \colon a \to b \text{ in $E$} \quad \Rightarrow \quad \overline {F_{a,b}}(u) = \overline{F_{a,b}}(v) \colon \overline{F_X}(a)
    \to \overline{F_X}(b) \text{ in $\D$}\text.
\end{equation*}
Here $\overline{F_{a,b}}$ denotes the unique extension of $F_{a,b}$ to a
function $\overline{\Sigma}(a,b) \to \D\big(\overline{F_X}(a),
\overline{F_X}(b)\big)$ commuting with the \smc\ term constructors.

\begin{example}\hfill
\begin{itemize}
\item An interpretation in $\D$ of the theory of monoids is a monoid in
    $\D$.
\item An interpretation in $\D$ of the theory of the linear lambda-calculus
    is given by an object $X \in \D$ and maps $\lambda \colon X \multimap X
    \to X$ and $@ \colon X \otimes X \to X$ rendering commutative the
    diagram
\begin{equation*}
    \cd{
      (X \multimap X) \otimes X \ar[r]^-{\lambda \otimes X}
      \ar[dr]_{\epsilon_{X,X}} &
      X \otimes X \ar[d]^{@} \\ & X\text.
    }
\end{equation*}
\end{itemize}
\end{example}
\begin{property}
To each syntactic theory $\mathbb T = (X, \Sigma, E)$ we may assign a small
\smc\ category $\C_\mathbb T$ which classifies $\mathbb T$, in the sense that
there is a bijection, natural in $\D$, between interpretations $\mathbb T \to
\D$ and strict \smc\ functors $\C_\mathbb T \to \D$.
\end{property}
\begin{proof}
We take the set of objects of $\C_{\mathbb T}$ to be $\overline X$, and obtain
its homsets by quotienting the sets $\overline \Sigma(a, b)$ under the smallest
congruence which contains each equation in $E$; makes composition associative
and unital; makes $\otimes$ and $\multimap$ functorial in each variable; makes
$\alpha$, $\lambda$, $\rho$, $\sigma$, $\epsilon$ and $\eta$ natural in each
variable; makes the $\lambda^{-1}$'s, $\rho^{-1}$'s and $\alpha^{-1}$'s inverse
to the $\lambda$'s, $\rho$'s and $\alpha$'s; verifies the triangle identities
for $\eta$ and $\epsilon$; and verifies the symmetric monoidal category axioms
of Mac~Lane.
\end{proof}

Observe that different syntactic theories $\mathbb T$ and $\mathbb T'$
may give rise to the same classifying category $\C_\mathbb T =
\C_{\mathbb T'}$, and so have the same models. Thus, in the spirit of
categorical logic, one should view syntactic $\smc$ theories as
\emph{presentations} of their classifying categories; so that to
understand a syntactic theory $\mathbb T$ is really to understand the
category $\C_\mathbb T$. The purpose of this note is to improve this
understanding by giving a \emph{graphical representation} of
$\C_{\mathbb T}$, in which morphisms are viewed as certain equivalence
classes of diagrams.  In the case where our theory has no equations,
and our signature no operations, we are considering a mere set of
types $X$, and the corresponding \smc\ category $\C_X$ is the free
\smc\ category on $X$.  We have mentioned that in this case we want to
use \citegenitifs{Hughes:freestar} representation.  We will show that
this special case suffices to derive the general one. In fact, it will
suffice to derive the case of a \emph{free theory}---one given by a
signature $(X, \Sigma)$ subject to no equations---since the
classifying category of an arbitrary theory may be obtained by
quotienting out the morphisms of the classifying category of a free
theory, so that a graphical representation of the latter induces a
graphical representation of the former.

Given a free theory $(X, \Sigma)$, we will obtain a graphical representation of
the corresponding classifying category $\C_{X, \Sigma}$ by first describing it
in terms of $\C_X$, the free \smc\ category on $X$, and then making use of a
suitable graphical description of the latter. We begin by introducing some
notation. We define the \emph{typing function} $\ty \colon \Sigma \to \ob \C_X
= \overline X$ by $\ty(\alpha) = s(\alpha) \multimap t(\alpha)$, and extend
this to a function on $\Sigma^\ast$, the set of lists in $\Sigma$, by taking
\begin{gather*}
    \ty() = I\text, \qquad \ty(\alpha) = s(\alpha) \multimap t(\alpha)\text, \\
    \text{and} \qquad \ty(\alpha_1, \dots, \alpha_n) = \ty(\alpha_1, \dots, \alpha_{n-1}) \tens
    \ty(\alpha_n) \text{ for $n \geqslant 2$.}
\end{gather*}
Though we may not have equality between $\ty(\alpha_1, \ldots, \alpha_n) \tens
\ty(\beta_1, \dots, \beta_m)$ and $\ty(\alpha_1, \dots, \alpha_n, \beta_1,
\dots, \beta_m)$, we can at least build a canonical isomorphism between them in
$\C_X$ using the associativity and unitality constraints. Similarly, for
$\sigma$ a permutation on $n$ letters, we can construct canonical maps
\begin{equation*}
    \hat \sigma \colon \ty(\alpha_1, \dots, \alpha_n) \to \ty(\alpha_{\sigma(1)}, \dots, \alpha_{\sigma(n)})
\end{equation*}
using the symmetry isomorphisms of $\C_X$. We now define a category
$\C'_{X,\Sigma}$ of which the classifying category $\C_{X, \Sigma}$ will be a
quotient.
\begin{itemize}
    \item \textbf{Objects} are objects of $\C_X$; \item \textbf{Morphisms}
        $U \to
	V$ are given by a list $\Gamma \in \Sigma^\ast$ together with a
morphism 	
\begin{equation*} 	    \phi \colon \ty(\Gamma) \tens U \to V 	\end{equation*} 	
in $\C_X$.
    \item \textbf{Identity maps} $U \to U$ are given by the empty list $()$
	together with the canonical isomorphism $I \tens U \to U$;
    \item \textbf{Composition} of maps $(\Gamma, \phi) \colon U \to V$ and
	$(\Delta, \psi) \colon V \to W$ is given by the map $(\Delta +
\Gamma,\, \xi) \colon U \to W$, wherein $\Delta + \Gamma$ is the
concatenation of the two lists, and 	$\xi$ is the composite morphism
	$$ \newdiagramgrid{rect1}{1.4,1}{1,1} 	
\begin{diagram}[grid=rect1,height=20pt,width=50pt]
	    \ty(\Delta + \Gamma) \tens U & & W \\
	    \dTo^{\cong} & & \uTo_{\psi} \\
	    \ty(\Delta) \tens \big(\ty(\Gamma) \tens U\big) & \rTo^{\ty(\Delta)\,
\tens\, \phi} & \ty(\Delta) \tens V 	\end{diagram} $$
\end{itemize}
The category $\C'_{X,\Sigma}$ admits an embedding functor $i \colon \C_X \to
\C'_{X,\Sigma}$, which is the identity on objects, and on morphisms sends a map
$\phi \colon U \to V$ to the pair of the empty list $()$ together with the
composite
$$ \begin{diagram}
    I \tens U & \rTo^{\cong} & U & \rTo^{\phi} V \text{.}
\end{diagram} $$
It also admits a tensor operation, which on objects is inherited from $\C_X$;
and on morphisms takes a pair of maps $(\Gamma, \phi) \colon U \to V$ and
$(\Gamma', \phi') \colon U' \to V'$ to the map $(\Gamma + \Gamma', \theta)
\colon U \tens U' \to V \tens V'$, where $\theta$ is the composite
$$ \begin{diagram}[width=20pt,height=15pt]
    \ty(\Gamma + \Gamma') \tens (U \tens U') & & & & V \tens V' \text{.} \\
    & \rdTo^{\cong} & & \ruTo^{\phi\, \tens\, \phi'} & \\
    & & (\ty(\Gamma) \tens U) \tens (\ty(\Gamma')\tens U') & &
\end{diagram} $$
However, this tensor operation does not underlie a tensor product in the usual
sense; for whilst functorial in each variable separately, it does not satisfy
the compatibility conditions required to obtain a functor of two variables.
These require the commutativity of squares of the form
\begin{equation} \label{squares}
    \begin{diagram}[width=40pt,height=20pt]
	U \tens U' & \rTo^{f\, \tens\, U'} & V \tens U' \\
	\dTo^{U\, \tens\, f'} & & \dTo_{V\, \tens\, f'}  \\
	U \tens V' & \rTo^{f\, \tens\, V'} & V \tens V' \text{;}
    \end{diagram}
\end{equation}
but we see from the definitions that, for $f = (\Gamma, \phi)$ and $f' =
(\Gamma', \phi')$ as above, the upper composite in~\eqref{squares} has its
first component given by $\Gamma' + \Gamma$, whilst the lower has it given by
$\Gamma + \Gamma'$; so that $\C'_{X,\Sigma}$ is not a \smc\ category.
Nonetheless, we do have that:

\begin{property}
    $\C'_{X,\Sigma}$ is a symmetric premonoidal category in the sense of
    \cite{power:premonoidal}, and the embedding $i \colon \C_X \to \C'_{X,\Sigma}$ is a strict symmetric premonoidal functor.
\end{property}
\begin{proof}
    Beyond the structure we have already noted, this means that $\C'_{X,\Sigma}$
    comes equipped with a unit object, which we take to be $I$, the unit object of
    $\C_X$; and with isomorphisms of associativity, unitality and symmetry of the
    same form as those for a symmetric monoidal category, but differing from them
    in two aspects. First, they need only be natural in each variable separately;
    so for symmetry, for instance, we only require diagrams of the following form
    to commute:
    $$ \begin{diagram}[width=35pt]
	U \tens V  & \rTo^{U\, \tens\, g} & U \tens V' \\
	\dTo~{\sigma_{U, V}} & & \dTo~{\sigma_{U, V'}} \\
	V \tens U & \rTo^{g\, \tens\, U} & V' \tens U
    \end{diagram}
    \quad \text{and} \quad
    \begin{diagram}[width=35pt]
	U \tens V & \rTo^{f\, \tens\, V} & U' \tens V  \\
	\dTo~{\sigma_{U, V}} & & \dTo~{\sigma_{U', V}} \\
	V \tens U & \rTo^{V\, \tens\, f} & V \tens U' \text{.}
    \end{diagram} $$
    Secondly, the constraint isomorphisms are required to be \emph{central} maps, where $f
    \colon U \to V$ is said to be \emph{central} just when for each $f' \colon U'
    \to V'$, the diagram \eqref{squares} and its dual
    $$ \begin{diagram}[width=40pt,height=20pt]
	U' \tens U & \rTo^{U'\, \tens\, f} & U' \tens V \\
	\dTo^{f'\, \tens\, U} & &  \dTo_{f'\, \tens\, V}  \\
	V' \tens U & \rTo^{V'\, \tens\, f} & V' \tens V
    \end{diagram} $$
    are rendered commutative.
    In the case of $\C'_{X,\Sigma}$, we fulfil these demands by taking each coherence
    constraint in $\C'_{X,\Sigma}$ to be the image of the corresponding coherence
    constraint in $\C_X$ under $i \colon \C_X \to \C'_{X,\Sigma}$. Naturality in each
    variable is easily checked; whilst centrality follows by observing that a map
    of $\C'_{X,\Sigma}$ is central iff it lies in the image of the aforementioned
    embedding.
    Finally, we observe that the embedding $i \colon \C_X \to \C'_{X,\Sigma}$ preserves
    all the structure of $\C_X$ on the nose, and sends central maps to central maps;
    and so is strict symmetric premonoidal.
\end{proof}

In fact, $\C'_{X,\Sigma}$ is \emph{closed} as a premonoidal category in the
sense that for each $V \in \C'_{X,\Sigma}$, the endofunctor $(\thg) \tens V$
has a right adjoint $V \multimap (\thg)$ which preserves central maps, with the
units and counits
\begin{equation*}
    U \ \longrightarrow \ V \multimap (U \tens V) \qquad \text{and} \qquad (V \multimap W) \tens V \ \longrightarrow \ W
\end{equation*}
of these adjunctions being central. Indeed, we may take the action of $V
\multimap (\thg)$ on objects to be given as in $\C_X$; and then we have:
\begin{align*}
    \C'_{X,\Sigma}(U \tens V,\, W) &= \Coprod{\Gamma \in \Sigma^\ast} \C_X\big(\ty(\Gamma) \tens (U \tens V),\, W\big) \\
    & \cong \Coprod{\Gamma \in \Sigma^\ast} \C_X\big(\ty(\Gamma) \tens U,\, V \multimap W\big) \\
    & = \C'_{X,\Sigma}(U,\, V \multimap W)\text,
\end{align*}
naturally in $U$ and $W$, as desired. The centrality requirements now amount to
the fact that the adjunctions $${(\thg) \tens V}\dashv\, {V \multimap (\thg)}
\colon \C'_{X,\Sigma} \to \C'_{X,\Sigma}$$ may be restricted and corestricted
to adjunctions $${(\thg) \tens V }\dashv\, {V \multimap (\thg)} \colon \C_X \to
\C_X.$$

The reason that $\C'_{X,\Sigma}$ is only premonoidal rather than monoidal is
that its morphisms are built from a \emph{list}, rather than a \emph{multiset}
of generating operations: in computational terms, we may think that a morphism
``remembers the order in which its generating operations are executed''. To
rectify this, we quotient out the morphisms of $\C'_{X,\Sigma}$ by the action
of the symmetric groups; the result will be the \smc\ category $\C_{X,\Sigma}$
we seek. So let there be given a list $\Gamma = (\alpha_1, \dots, \alpha_n) \in
\Sigma^\ast$, a permutation $\sigma \in S_n$, and a morphism $\phi \colon
\ty(\sigma \Gamma) \tens U \to V$ in $\C_X$, where $\sigma \Gamma$ is the list
$(\alpha_{\sigma(1)}, \dots, \alpha_{\sigma(n)})$. A generating element for our
congruence $\sim$ on the morphisms of $\C'_{X,\Sigma}$ is now given by
\begin{equation*}
    (\sigma\Gamma,\, \phi) \quad \sim \quad (\Gamma,\, \phi \circ (\hat \sigma \tens U))
\end{equation*}
where we recall that $\hat \sigma$ is the canonical morphism $\ty(\Gamma) \to
\ty(\sigma \Gamma)$ built from symmetry and associativity maps in $\C_X$. We
may now verify that for morphisms
$$ \begin{diagram}
    U & \rTo^{f} & V & \pile{\rTo^{g} \\ \rTo_{h}} & W & \rTo^{k} & Z
\end{diagram} $$
in $\C'_{X,\Sigma}$, $g \sim h$ implies both $gf \sim hf$ and $kg \sim kh$, so
that $\sim$ is a congruence on $\C'_{X,\Sigma}$, and we may define the category
$\C_{X,\Sigma}$ to be the quotient of $\C'_{X,\Sigma}$ by $\sim$.

\begin{property}
    $\C_{X,\Sigma}$ is a symmetric monoidal closed category, and the quotient map $q
    \colon \C'_{X,\Sigma} \to \C_{X,\Sigma}$ is a strict symmetric premonoidal
    functor.
\end{property}
\begin{proof}
    Straightforward checking shows that if $f \sim f'$ and $g \sim g'$ in $\C'_{X,\Sigma}$, then $f \tens g \sim f' \tens g'$, so that the tensor operation
    on $\C'_{X,\Sigma}$ passes to the quotient $\C_{X,\Sigma}$. For this operation to
    define a bifunctor on $\C_{X,\Sigma}$, we must verify that squares of the form
    \eqref{squares} commute in $\C_{X,\Sigma}$: and this follows by checking that
    \begin{equation*}
	(f \tens V') \circ (U \tens f') \sim (V \tens f') \circ (f \tens U')
    \end{equation*}
    in $\C'_{X,\Sigma}$. This defines our binary tensor on $\C_{X,\Sigma}$; whilst the
    nullary tensor we inherit from $\C'_{X,\Sigma}$.
    The associativity, unitality and symmetry constraints in the category $\C_{X,\Sigma}$ are
    obtained as the image of the corresponding constraints in $\C'_{X,\Sigma}$ under
    the quotient map. Commutativity of the triangle, pentagon and hexagon axioms is
    inherited; whilst the (restricted) naturality of these maps in $\C'_{X,\Sigma}$
    becomes their (full) naturality in $\C_{X,\Sigma}$. Thus $\C_{X,\Sigma}$ is
    symmetric monoidal.
    It is now easy to check that the isomorphisms $\C'_{X,\Sigma}(U \tens V,\, W)
    \cong \C'_{X,\Sigma}(U,\, V \multimap W)$ descend along the quotient map,
    and so induce a closed structure on $\C_{X,\Sigma}$. Finally, since each piece
    of structure on $\C_{X,\Sigma}$ is obtained from the corresponding piece of
    structure on $\C'_{X,\Sigma}$, the quotient map $q \colon \C'_{X,\Sigma} \to \C_{X,\Sigma}$ is strict symmetric premonoidal as required.
\end{proof}
Observe that the composite functor $qi \colon \C_X \to \C_{X,\Sigma}$, is a
strict symmetric premonoidal closed functor between two symmetric monoidal
closed categories; and as such, is actually a strict symmetric \emph{monoidal}
closed functor. We make use of this fact below.

\begin{theorem}\label{thm:coherence}
    $\C_{X,\Sigma}$ is the classifying category of the syntactic theory with
    signature $(X, \Sigma)$ and no equations.
\end{theorem}
\begin{proof}
    Suppose first given a strict \smc\ functor $F \colon \C_{X,\Sigma} \to \D$; we obtain an
    interpretation $G \colon (X, \Sigma) \to \D$ by taking
    \begin{equation}\label{eqs}
        G_X(x) = F(x) \quad \text{and} \quad G_{a,b}(\alpha) = F[\alpha] \colon Fa \to
        Fb\text,
    \end{equation}
    where, for $\alpha \in \Sigma(a,b)$, the morphism $[\alpha] \colon a \to b$  of $\C_{X, \Sigma}$ is given by $q \big((\alpha),
    \epsilon_{ab}\big)$.
    Conversely, we must show that each interpretation $G \colon (X, \Sigma) \to \D$ lifts to a unique strict \smc\ functor $F \colon \C_{X,\Sigma} \to \D$
    satisfying \eqref{eqs}. The action of $G$ on ground types is given by a function $G_X \colon X \to \ob
    \D$; and this is equally well a functor $G_X \colon X \to \D$---with $X$ regarded now as a
    discrete category---which, as $\C_X$ is the free \smc\ category on $X$, lifts to a strict \smc\ functor
    $\tilde G_X \colon \C_X \to \D$.
    It follows that $F$, if it exists, must makes the following diagram of strict \smc\
    functors commute:
    \begin{equation}\label{newt}
	\begin{diagram}[nohug,width=40pt,height=20pt]
	    & & \C_{X,\Sigma} & & \\
	    & \ruTo^{qi} & & \rdTo^{F} & \\
	    \C_X & & \rTo^{\tilde G_X} & & \C\text. 	\end{diagram}
    \end{equation}
    Indeed, to ask that the first equation in~\eqref{eqs} should hold
    is equally well to ask that~\eqref{newt} should commute when precomposed with
    the functor $\eta \colon X \to \C_X$ exhibiting $\C_X$ as free on $X$; and
    by the uniqueness part of the universal property of $\C_X$, this is equally
    well to ask~\eqref{newt} itself to commute. This determines the action of $F$ on
    the objects and certain of the morphisms of $\C_{X,\Sigma}$; let us now extend
    this to deal with an arbitrary morphism $f \colon U \to V$. If $f$ is
    represented by some $(\Gamma, \phi)$ in $\C'_{X,\Sigma}$, then we may factorise
    it as
    $$
    \newdiagramgrid{line8}{1,1.5,1.5,1}{}
    \begin{diagram}[grid=line8,width=30pt]
	U & \rTo^{q(\Gamma, \id)} & \ty(\Gamma) \tens U & \rTo^{q(i(\phi))} & V
    \end{diagram} $$
    in $\C_{X,\Sigma}$; and commutativity in \eqref{newt} forces $F$, if it exists, to send the
    second part of this factorisation to $\tilde G_X(\phi)$. For the first
    part, either we have $\Gamma$ empty, in which case $q(\Gamma, id)$ is the unit
    isomorphism $U \cong I \tens U$; or we have $\Gamma = (\alpha_1, \dots,
    \alpha_n)$, in which case $q(\Gamma, \id)$ decomposes as 
    $$
    \newdiagramgrid{line9}{0.9,1.8,2.8,2.5}{}
    \begin{diagram}[grid=line9,width=25pt]
	U & \rTo^{\cong} & (\Bigtens{1 \leq i \leq n} I\ ) \tens U &
\rTo^{\bigtens_i \overline{[\alpha_i]}\, \tens\, U} & \ty(\Gamma) \tens U
\text{,}
    \end{diagram} $$
    where $\overline{[\alpha_i]} \colon I \to \ty(\alpha_i)$ is the
    exponential transpose of $[\alpha_i] \colon s(\alpha_i) \to t(\alpha_i)$ in $\C_{X, \Sigma}$. But since we require $F$, if it exists, to both satisfy
    the second equation in~\eqref{eqs} and strictly preserve the \smc\ structure, this
    determines its value on $q(\Gamma, \id)$; and hence on an arbitrary morphism of $\C_{X,\Sigma}$.
    Consequently, there is at most one strict \smc\ functor
    $F \colon \C_{X,\Sigma} \to \C$
    satisfying the equations in~\eqref{eqs}; and
    in order to conclude that there is \emph{exactly} one such, we must check that
    the assignations described above underlie a well-defined strict \smc\ functor $F$. This follows by
    straightforward calculation: as a representative sample of which, we verify that $F$
    as given above is well-defined on morphisms.
    So let there be given $f \colon U \to V$ in $\C_{X,\Sigma}$, together with two
    morphisms $(\sigma\Gamma,\, \phi)$ and $(\Gamma,\, \phi \circ (\hat \sigma
    \tens U))$ of $\C'_{X,\Sigma}$ which represent it. Then we have the following
    commutative diagram in $\C_{X,\Sigma}$:
    $$ \begin{diagram}[nohug,width=40pt,height=15pt]
	& & \ty(\Gamma) \tens U  & &\\
	& \ruTo^{q(\Gamma, \id)} & & \rdTo^{qi(\phi \circ (\hat \sigma \tens U))} & \\
	U & & \dTo~{qi(\hat \sigma \tens U)} & & V \\
	& \rdTo_{q(\sigma\Gamma, \id)} & & \ruTo_{qi(\phi)} & \\
	& & \ty(\sigma \Gamma) \tens U & &
    \end{diagram} $$
    and must show that the corresponding diagram commutes when we apply
    $F$. This is clear for the right-hand triangle; whilst for the
    left-hand one, it amounts to checking the following equality in $\D$:
    $$
    \newdiagramgrid{pent}{1,1,1,1}{0.5,0.5,1,1}
    \begin{diagram}[grid=pent,nohug,width=35pt]
	& & FU & & \\
	& \ldTo^{\cong} & & \rdTo^{\cong} & \\
	(\Bigtens{1 \leq i \leq n} I\ ) \tens FU & & & & (\Bigtens{1 \leq i \leq n} I\ ) \tens FU \\
	\dTo~{\bigtens_i \overline{[\alpha_i]}\, \tens\, FU} & & & & \dTo~{\bigtens_i \overline{[\alpha_{\sigma i}]}\, \tens\, FU} \\
	F(\ty(\Gamma)) \tens FU & & \rTo^{F \hat \sigma\, \tens\, FU} & &
F(\ty(\sigma \Gamma)) \tens FU \text{,}
    \end{diagram} $$
    which follows immediately from the symmetric monoidal closed category
    axioms.    The remaining calculations proceed similarly.
\end{proof}

Finally in this section, we consider the case of a general theory $\mathbb{T} = (X, \Sigma, E)$. 
Let $\C_\theory$ be the quotient of $\C_{X,\Sigma}$ by the smallest congruence $\sim$ 
which contains all the equations in $E$ and respects the $\smc$ structure.
We have:
\begin{theorem}
      $\C_{\mathbb{T}}$ is the classifying category of the theory $\mathbb{T}$.
\end{theorem}

In fact, using a linear analogue of \citegenitif{LambekScott}
\emph{functional completeness}, we may give a more direct characterisation of the congruence $\sim$.
Here we write $\name f \colon I \to a \impll b$ to denote the currying of any map $f \colon a \to b$.
\begin{property}\label{lemma:fcomp} We obtain $\sim$ as the smallest equivalence relation generated by $\sim_1$, where
  $f \sim_1 g \colon a \to b$ just when there exists an equation $u = v
  \colon c \to d$ in $E$ and map $h$ such that $f$ is
\begin{diagram}
    a & \rTo^{\iso} & I \tens a & \rTo^{\name{u} \tens a} & (c \impll d)
    \tens a & \rTo^{h} & b
  \end{diagram}
  and replacing $u$ with $v$ yields $g$.
\end{property}

\section{A graphical representation of the classifying category}
Putting Theorem~\ref{thm:coherence} together with
\citet{Hughes:freestar}'s graphical description of $\C_X$, we obtain
the following graphical representation of the category $\C_{X,
  \Sigma}$. First, for each type $a \in \overline X$, we define the
\emph{ports} of $a$ to be the set of leaf occurrences in it, which may
either be of type $I$, or of ground types $x \in X$. Ports are signed
\emph{positive} when they are reached by passing to the left of an
even number of $\multimap$, and \emph{negative} otherwise. We let
$a^+$ and $a^-$ denote the sets of positive and negative ports of $a$,
respectively. We define a \emph{support} to be a finite set labelled
by elements of $\Sigma$. The \emph{ports} of a support $C$ are defined
by
\begin{center}
    $C^+ = \coprod_{c \in C} (\ty (\alpha_c))^+$
    \hfil and \hfil
    $C^- = \coprod_{c \in C} (\ty (\alpha_c))^-,$
\end{center}
where $\alpha_c$ is the label of $c$.  We now define the category 
$\D^0_{X, \Sigma}$ of $(X, \Sigma)$-\emph{prenets} to have:
\begin{itemize}
    \item \textbf{Objects} being elements of $\overline X$.
    \item \textbf{Morphisms} $a \to b$ being given by a support $C$
        together with a directed graph $G$, whose vertices are the disjoint
        union of the
ports of $a$, $b$ and $C$; and whose edges are such that the incidence
relation is the graph of a partial function
	\begin{equation}
	    g \colon a^+ + C^+ + b^- \rightharpoonup a^- + C^- + b^+,\label{eq:pleafun}
	\end{equation}
	that restricts to a bijection of $x$-labeled ports for each $x
        \in X$.  We consider morphisms equivalent up to the choice of
        support (replacing $C$ with isomorphic $C'$, preserving $g$).
    \item \textbf{Identity maps} $a \to a$ being given by the empty support
        together with the identity graph.
    \item \textbf{Composition} of maps $(C, G) \colon a \to b$ and $(D, H)
	\colon b \to c$ being given by the map $(C + D,\, G +_b H)
	\colon a \to c$, where $G +_b H$ is obtained by glueing the graphs $G$
and $H$ together along the ports of $b$. More formally, if $x \in G$ and $z
\in H$, then $G +_b H$ will have an edge $x \to z$ whenever there exist
ports $y_1, \dots, y_k$ of $b$ and edges
\begin{equation*}
\cd[@-1em@R-1em]{
    x \ar[r] & y_1 & y_2 \ar[r] & y_3 & \ \ \ \dots \ \ \ & y_{k-1} \ar[r] & y_k & & \text{in $G$}\\
    \text{and} & y_1 \ar[r] & y_2 & y_3 \ar[r] &  \ \ \ \dots \ \ \ \ar[r] & y_{k-1} & y_k \ar[r] & z & \text{in $H$.}\\
}
\end{equation*}
There are three analogous cases when:
\begin{itemize}\item
  $x \in H$ and $z \in G$,
\item $x, z \in H$, or
\item $x, z \in G$.
\end{itemize}
\end{itemize}

We now consider the subcategory $\D^1_{X, \Sigma}$ of $(X,
\Sigma)$-\emph{nets} with the same objects, but whose morphisms are
\emph{correct} prenets in the following sense.  First, for any \imll\
formula $a$, let $a'$ be its representation in classical MLL, i.e.,
using $\tens$, $\parr$, $I$, $\bot$, and signed ground types $x$ and
$\nop{x}$; in particular, $(a \impll b)' = \nop{a'}
\parr b'$. Now by a \emph{switching} of $a'$, we mean a graph obtained by cutting
exactly one premise of each $\parr$ node in the abstract syntax tree of $a'$;
and by a \emph{switching} of a $(X, \Sigma)$-prenet $(C, G) \colon a \to b$, we
mean a graph obtained by gluing along the ports:
\begin{itemize}
\item A switching of $\nop{a'}$;
\item A switching of $b'$;
\item A switching of each $\nop{\alpha_c'}$ (where $\alpha_c$ is the label
    of $c \in C$); and
\item The graph $G$ (forgetting the orientation).
\end{itemize}
The prenet $(C, G)$ is said to be \emph{correct}, or a \emph{net}, just when
all its switchings are trees. The nets $a \to b$ are in close correspondence
with the morphisms $a \to b$ in the free \smc\ category $\C_{X, \Sigma}$. To
see this, suppose given a net $(C, G) \colon a \to b$ whose support is a finite
set $\{1, \dots, n\}$. If we define $\Gamma = (\alpha_1, \ldots, \alpha_n)$
then we have $C^+ = \coprod_{1 \leq i
  \leq n} (\ty (\alpha_i))^+ \iso (\ty (\Gamma))^+$ and $C^-
\iso (\ty (\Gamma))^-$; and we claim that the composite partial function
\begin{diagram}
  (\ty (\Gamma) \tens a)^+ + b^- & \rTo^{\iso} & a^+ + C^+ + b^- & \\
  \dDashto<{g'}  & & \dTo>{g} \\
  (\ty (\Gamma) \tens a)^- + b^+ & \lTo^{\iso} & a^- + C^- + b^+.
\end{diagram}
describes a morphism $\ty(\Gamma) \tens a \to b$ in
\citet{Hughes:freestar}'s presentation of the free \smc\ category
$\C_X$ over $X$. For this, we just have to show correctness; but any
switching of $\nop{(\ty (\Gamma) \tens a)}$ amounts to a disjoint
union of a switching of each of $\nop{a}$ and the $\nop{\ty
  (\alpha_i)}$'s, so that correctness follows from that of $g$. Thus
$(C,G)$ yields a morphism $a \to b$ in $\C_{X, \Sigma}$; and
conversely, given $\Gamma$, any correct representative $g'$ in the
sense of Hughes defines a correct net in our sense, with reordering of
$\Gamma$ resulting in an isomorphism of supports.

Finally, we may mimic Trimble rewiring in our setting: say that $f \sim g$ when
$g$ is obtained by changing the target of a single edge from a negative
occurrence of $I$ in $f$, preserving correctness. This extends to an
equivalence relation which we call \emph{rewiring}. Letting $\D_{X, \Sigma}$ be
the quotient of  $\D^1_{X, \Sigma}$ modulo rewiring, we obtain:
\begin{theorem}
  The categories $\D_{X, \Sigma}$ and $\C_{X, \Sigma}$ are isomorphic in $\SMCCat$.\end{theorem}

The category $\D_{X, \Sigma}$ provides a graphical representation
of the free \smc\ category generated by $(X, \Sigma)$. 
If $X = \{x,y\}$ and $\Sigma$ is described by the following graph:
\begin{diagram}
    x & \rTo^{\alpha} & x \tens y & & y \tens (x \impll y) & \rTo^{\beta} & y
\end{diagram}
then an example morphism from $x \tens ((x \tens I) \impll y)$ to $I \impll (x \tens y)$ of $D_{X,\Sigma}$ is:
\begin{center}
    \includegraphics{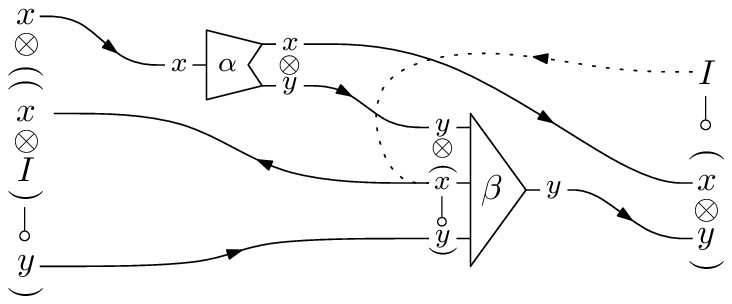}
\end{center}
Notice that the dotted link can be rewired to any positive port. 

\bibliographystyle{plainnat}
\bibliography{bib}

\begin{thebibliography}{15}
\providecommand{\natexlab}[1]{#1}
\providecommand{\url}[1]{\texttt{#1}}
\expandafter\ifx\csname urlstyle\endcsname\relax
  \providecommand{\doi}[1]{doi: #1}\else
  \providecommand{\doi}{doi: \begingroup \urlstyle{rm}\Url}\fi

\bibitem[Blute et~al.(1996)Blute, Cockett, Seely, and Trimble]{Cockett:weakly}
R.~Blute, J.~R.~B. Cockett, R.~A.~G. Seely, and T.~H. Trimble.
\newblock Natural deduction and coherence for weakly distributive categories.
\newblock \emph{Journal of Pure and Applied Algebra}, 13\penalty0 (3):\penalty0
  229--296, 1996.

\bibitem[Cheng(2003)]{Cheng:trees}
Eugenia Cheng.
\newblock {A relationship between trees and Kelly-Mac Lane graphs}.
\newblock ArXiv Mathematics e-prints, math/0304287, April 2003.

\bibitem[Danos and Regnier(1989)]{Danos:mll}
Vincent Danos and Laurent Regnier.
\newblock The structure of multiplicatives.
\newblock \emph{Archive for Mathematical Logic}, 28:\penalty0 181--203, 1989.

\bibitem[Girard(1993)]{Girard:survey}
Jean-Yves Girard.
\newblock Linear logic: a survey.
\newblock In \emph{Proc. International Summer School of Marktoberdorf}, F94,
  pages 63--112. NATO Advanced Science Institute, 1993.

\bibitem[Girard(1987)]{ll}
Jean-Yves Girard.
\newblock Linear logic.
\newblock \emph{Theoretical Computer Science}, 50:\penalty0 1--102, 1987.

\bibitem[Hughes(2005)]{Hughes:freestar}
Dominic J.~D. Hughes.
\newblock Simple free star-autonomous categories and full coherence.
\newblock ArXiv Mathematics e-prints, math/0506521, June 2005.

\bibitem[Jensen and Milner(2004)]{Milner:bigraphs}
Ole~H. Jensen and Robin Milner.
\newblock Bigraphs and mobile processes (revised).
\newblock Technical Report TR580, University of Cambridge, 2004.
\newblock URL \url{http://www.cl.cam.ac.uk/TechReports/UCAM-CL-TR-580.pdf}.

\bibitem[Kelly and {Mac Lane}(1971)]{Kelly:coherence}
{G. M.} Kelly and Saunders {Mac Lane}.
\newblock Coherence in closed categories.
\newblock \emph{Journal of Pure and Applied Algebra}, 1\penalty0 (1):\penalty0
  97--140, 1971.

\bibitem[Lamarche and Strassburger(2006)]{Lamarche:nets}
Fran{\c{c}}ois Lamarche and Lutz Strassburger.
\newblock From proof nets to the free *-autonomous category.
\newblock \emph{Logical Methods in Computer Science}, 2\penalty0 (4), 2006.

\bibitem[Lambek and Scott(1988)]{LambekScott}
J.~Lambek and {P. J.} Scott.
\newblock \emph{Introduction to Higher-Order Categorical Logic}.
\newblock Cambridge University Press, 1988.

\bibitem[Lawvere(1963)]{Lawvere:thesis}
F.~W. Lawvere.
\newblock \emph{Functorial semantics of algebraic theories}.
\newblock PhD thesis, Columbia University, 1963.

\bibitem[Pfenning and Elliott(1988)]{PfenningElliott:hoas}
Frank Pfenning and Conal Elliott.
\newblock Higher-order abstract syntax.
\newblock In \emph{ACM SIGPLAN '88 Symposium on Language Design and
  Implementation}, pages 199--208. ACM, 1988.

\bibitem[Power and Robinson(1997)]{power:premonoidal}
John Power and Edmund Robinson.
\newblock Premonoidal categories and notions of computation.
\newblock \emph{Mathematical. Structures in Comp. Sci.}, 7\penalty0
  (5):\penalty0 453--468, 1997.

\bibitem[Tan(1997)]{Tan:phd}
{Audrey M.} Tan.
\newblock \emph{Full Completeness for Models of Linear Logic}.
\newblock PhD thesis, University of Cambridge, 1997.

\bibitem[Trimble(1994)]{Trimble:phd}
{Todd H.} Trimble.
\newblock \emph{Linear logic, bimodules, and full coherence for autonomous
  categories}.
\newblock PhD thesis, Rutgers University, 1994.

\end{thebibliography}

\end{document}